\documentclass[11pt,english]{article}
\usepackage[margin=1in]{geometry}

\usepackage[english]{babel}
\usepackage[utf8]{inputenc}
\usepackage[T1]{fontenc}

\usepackage{amsmath}
\usepackage{amsthm,amssymb}
\usepackage{graphicx}
\usepackage[colorlinks=true, allcolors=blue]{hyperref}
\usepackage{xcolor}
\usepackage{cite}
\usepackage{cleveref}
\usepackage{tikz}
\usepackage{thmtools}
\usepackage{thm-restate}
\usetikzlibrary{arrows.meta,arrows}

\newtheorem{theorem}{Theorem}
\newtheorem{lemma}[theorem]{Lemma}
\newtheorem{definition}[theorem]{Definition}

\newtheorem{conjecture}[theorem]{Conjecture}
\newtheorem{hypothesis}[theorem]{Hypothesis}
\newtheorem{corollary}[theorem]{Corollary}

\newcommand{\p}{\mathsf{P}}
\newcommand{\q}{\mathsf{Q}}
\newcommand{\Cmax}{C_{\mathsf{max}}}
\newcommand{\Prec}{\mathsf{prec}}
\newcommand{\dup}{\mathsf{dup}}

\title{On the Hardness of Scheduling With Non-Uniform Communication Delays}

\author{Sami Davies\thanks{University of Washington, Seattle. Email: {\tt\{daviess,rothvoss,yihaoz93\}@uw.edu}. Sami Davies was an intern at MSR, Redmond when this work was done. Thomas Rothvoss is supported by NSF CAREER grant 1651861 and a David \& Lucile Packard Foundation Fellowship. } \qquad Janardhan Kulkarni\thanks{Microsoft Research, Redmond. Email: {\tt\{jakul,jatarnaw\}@microsoft.com}.} \qquad Thomas Rothvoss\footnotemark[1] \\ Sai Sandeep\thanks{Carnegie Mellon University, {\tt spallerl@andrew.cmu.edu}. Research supported in part by NSF grants CCF-1563742 and CCF-1908125.} \qquad 
Jakub Tarnawski\footnotemark[2] \qquad Yihao Zhang\footnotemark[1]}

\date{}

\begin{document}

\maketitle

\begin{abstract}
    In the scheduling with non-uniform communication delay problem, the input is a set of jobs with precedence constraints. Associated with every precedence constraint between a pair of jobs is a communication delay, the time duration the scheduler has to wait between the two jobs if they are scheduled on different machines. The objective is to assign the jobs to machines to minimize the makespan of the schedule. 
    Despite being a fundamental problem in theory and a consequential problem in practice, the approximability of scheduling problems with communication delays is not very well understood.
    One of the top ten open problems in scheduling theory, in the influential list by Schuurman and Woeginger and its latest update by Bansal, asks if the problem admits a constant factor approximation algorithm. In this paper, we answer the question in negative by proving that there is a logarithmic hardness for the problem under the standard complexity theory assumption that NP-complete problems do not admit quasi-polynomial time algorithms.
    
    Our hardness result is obtained using a surprisingly simple reduction from a problem that we call Unique Machine Precedence constraints Scheduling (UMPS). We believe that this problem is of central importance in understanding the hardness of many scheduling problems and conjecture that it is very hard to approximate. Among other things, our conjecture implies a logarithmic hardness of related machine scheduling with precedences, a long-standing open problem in scheduling theory and approximation algorithms. 
\end{abstract}

\section{Introduction}
\label{sec:intro}

We study the problem of scheduling jobs with precedence and non-uniform communication delay constraints on identical machines to minimize makespan objective function.
This classic model was first introduced by Rayward-Smith~\cite{RAYWARDSMITH1987} and Papadimitriou and Yannakakis~\cite{PapadimitriouY90}.
In this problem, we are given a set  $J$ of $n$ jobs, where each job $j$ has a processing length $p_j  \in \mathbb{Z}_+$.
The jobs need to be scheduled on $m$ identical machines.
The jobs have precedence and communication delay constraints, which are given by a partial order $\prec$. 
A constraint $j \prec j'$ encodes that job $j'$ can only start after job $j$ is completed.
Moreover, if $j \prec j'$ and $j$, $j'$ are scheduled on {different machines}, then $j'$ can only start executing at least $c_{jj'}$ time units after $j$ had finished.
On the other hand,  if $j$ and $j'$ are scheduled on the {same machine}, then $j'$ can start executing immediately after $j$ finishes. 
The goal is to schedule jobs {\em non-preemptively} to minimize makespan objective function, which is defined as the completion time of the last job.
In a non-preemptive schedule, each job $j$ needs to be assigned to a single machine $i$ and executed during a contiguous time interval of length $p_j$.
In the classical scheduling notation, the problem is denoted by $\p  \mid \Prec, c_{jk}\mid \Cmax$.\footnote{
We adopt the convention of \cite{GLLR79,VeltmanLL90}, where the respective fields denote:
	\textbf{(1) machine environment:} $Q$ for related machines, $P$ for identical machines,
	\textbf{(2) job properties:} $\Prec$ for precedence constraints; $c_{jk}$ for communication delays; $c$ when all the communication delays are equal to $c$; $p_j = 1$ for unit length case,
	\textbf{(3) objective:} $\Cmax$ for minimizing makespan.
}
A closely related problem is $\p \infty \mid \Prec, c_{jk}\mid \Cmax$, where the scheduler has access to as many machines as required.

Scheduling jobs with precedence and communication delays has been studied extensively over many years \cite{RAYWARDSMITH1987,PapadimitriouY90,MunierKonig,HanenMunier73Apx,ThurimellaYesha,HoogeveenLV94,PapadimitriouY90,GiroudeauKMP08}.
Furthermore, due to its relevance in data center scheduling problems and large scale training of ML models, there has been a renewed interest in more applied communities; we refer the readers to ~\cite{Chowdhury, guo2012spotting,hong2012finishing,shymyrbay2018meeting,zhang2012optimizing,zhao2015rapier,luo2016towards, narayanan2018pipedream, dean2017rlplacement,spotlight2018,jia2018beyond,tarnawski2020efficient}. 
However from a theoretical standpoint, besides NP-hardness results, very little was known in terms of the algorithms for the problem until the recent work by Maiti et al.~\cite{MRSSV} and Davies et al.~\cite{SchedulingWithCommunicationDelaysDKRTZ2020, DaviesKRTZ21}.
These very recent papers designed polylogarithmic approximation algorithms for the special case when all the communication delays are the {\em same}.
We survey these results in Section \ref{sec:relatedwork}.
In fact, the problems of scheduling jobs with communication delays are some of the well-known open questions in approximation algorithms and scheduling theory, which have resisted progress for a long time.  For this reason, the influential survey by  Schuurman and Woeginger~\cite{SW99a} and its recent update by Bansal~\cite{Bansalmapsp} list understanding the approximability of the problems in this space as one of the top-10 open questions in scheduling theory. 

In particular, an open problem in these surveys asks if the non-uniform communication delay problem on identical machines, even assuming an unbounded number of machines ($\p \infty \mid \Prec, c_{jk}\mid \Cmax$), admits a constant factor approximation algorithm.
The main result of the paper resolves this question.

\begin{restatable}{theorem}{main}
\label{thm:main}
For every constant $\epsilon >0$, assuming $\emph{\textsf{NP}}\nsubseteq \emph{\textsf{ZTIME}}\left(n^{(\log n)^{O(1)}}\right)$, the non-uniform communication delay problem $(\p \infty \mid \Prec,  c_{jk}\mid \Cmax)$ does not admit a polynomial time $(\log n)^{1-\epsilon}$ approximation algorithm. 
\end{restatable}

As $\p \infty \mid \Prec, c_{jk}\mid \Cmax$, the problem with an unbounded number of machines is a special case of $\p \mid \Prec, c_{jk}\mid \Cmax$, where the number of machines is specified, our theorem also implies same hardness for $\p \mid \Prec, c_{jk}\mid \Cmax$.

\subsection{Our Techniques}

Our hardness result is obtained using a reduction via a problem we call Unique Machines Precedence constraints Scheduling (UMPS). In this problem, there are  $m$ machines and $n$ jobs $j_1, j_2, \ldots, j_n$ with precedence relations among them. Each job $j_l$ has length $p(l)$ and can be scheduled only on a {\em unique} machine $M(l) \in [m]$. The objective is to schedule the jobs non-preemptively on the corresponding unique machines respecting the precedence relations to minimize the makespan objective function. Our proof of Theorem \ref{thm:main} proceeds via two steps:

\begin{enumerate}
    \item First we show a reduction from an instance $I$ of the UMPS problem to an instance $I'$ of the non-uniform communication delay problem. 
    The key observation is to make sure that the set of jobs $J(i)$ that need to be scheduled on machine $i$ in $I$ do not get scheduled on multiple machines in $I'$.
    This we achieve by introducing a dummy job $j^*_i$ and introducing precedence constraints from all jobs in $J(i)$ to $j^*_i$ and a very large communication delay. 
    This ensures that $J(i)$ and $j^*_i$ are scheduled on the same machine in $I'$, although this machine need not be $i$.
    Despite this, we show that any valid schedule of $I'$ can be mapped back to a feasible schedule of $I$, with almost the same makespan.
    Our reduction creates only two types of communication delays and works for the unit length case.
    
    \item Next we observe that the UMPS problem generalizes the classical job shops problem (see e.g.,\cite{lawler1993sequencing,LeightonMR94, MastrolilliS11}), whose approximation is well understood~\cite{ShmoysSW94, CzumajS00, GoldbergPSS01, FeigeS02}. The logarithmic hardness result of the acyclic job shops problem by Mastrolli and Svensson ~\cite{MastrolilliS11} implies a logarithmic hardness of the UMPS problem. We remark that the hardness result of ~\cite{MastrolilliS11} only works when jobs have lengths, and hence our Theorem \ref{thm:main} only applies to the setting where jobs have lengths.
\end{enumerate}

\begin{figure}
\centering
\label{fig:reduction}
\begin{tikzpicture}

\node(A)[draw, thick] at (2,-2) {Job shops};
\node(B)[draw, thick, text width=4.5cm] at (7,-2) {Unique Machine Precedence constraints Scheduling (UMPS)};
\node(C)[draw, thick, text width=4.5cm] at (13,-2) {Scheduling with non-uniform communication delay};

\draw [-{Latex[length=3mm,width=3mm]}, thick] (A) -- (B);
\draw[-{Latex[length=3mm,width=3mm]}, thick] (B) -- (C);
\end{tikzpicture}
\caption{Role of UMPS problem in our hardness reduction.}

\end{figure}
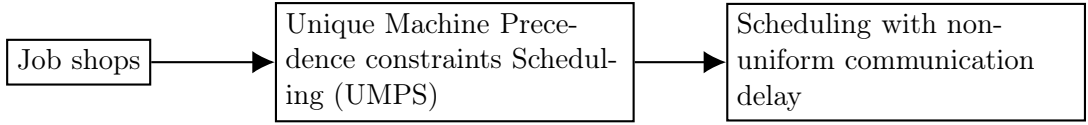

In hindsight, our proof of Theorem \ref{thm:main} is surprisingly simple. 
However, the main conceptual contribution of our proof is in identifying the UMPS problem as a central problem to study, which has implications for the hardness of various scheduling problems.
Furthermore, the UMPS problem, which can be viewed as a generalization of job shop scheduling model or as a highly restricted version of multiple dimensional scheduling problem with precedences, or as a restricted assignment problem with precedence constraints, is a fundamental problem to study on its own, both from a theoretical perspective and also from a practical point of view.
We believe the UMPS problem is a key intermediate step towards resolving several long-standing open problems in scheduling theory.
We make the following two conjectures regarding the approximability of the UMPS problem.

\begin{conjecture}
\label{conj:dag-main1}
There exists a constant $\epsilon < 1$ such that it is NP-hard to approximate UMPS containing $n$ jobs within a factor of $n^\epsilon$, even when all jobs have unit lengths. 
\end{conjecture}

\begin{conjecture}
\label{conj:dag-main}
There exists an absolute constant $C\geq 1$ such that the following holds. For every constant $\epsilon>0$, it is NP-hard to approximate UMPS containing $n$ jobs within a $(\log n)^{1-\epsilon}$ factor, even when the number of machines $m$ is at most $(\log n)^{C}$ and all the jobs have unit lengths. 
\end{conjecture}


Our second main contribution of the paper is to show that the above conjectures imply hardness results for various problems.
In particular, Conjecture \ref{conj:dag-main} implies logarithmic hardness for scheduling with precedences on related machines, another top-10 problem in scheduling theory \cite{SW99a, Bansalmapsp} and in the approximation algorithm book of Shmoys and Williamson \cite{theapproxbook}. 

\begin{restatable}{theorem}{related}
\label{thm:related}
Assuming~\Cref{conj:dag-main} and $\textsf{NP}\nsubseteq \textsf{DTIME}\left(n^{(\log n)^{O(1)}}\right)$, there exists an absolute constant $\gamma>0$ such that scheduling related machines with precedences problem $(\q \mid \Prec \mid \Cmax)$ has no polynomial time $\Omega\left( (\log m)^{\gamma}\right)$ factor approximation algorithm. 
\end{restatable}

Previously, Bazzi and Norouzi-Fard~\cite{BazziN15} introduced a $k$-partite hypergraph partition problem whose hardness implies a super constant hardness of scheduling with precedences on related machines. Our reduction uses the same idea of job replication from the work of \cite{BazziN15}, while our soundness analysis is technically more involved. 
We also show that the hypothesis of~\cite{BazziN15} 
implies a super constant hardness of the UMPS problem. Thus, our problem can be viewed as a weaker version of the hypothesis of~\cite{BazziN15} with the same implication towards the hardness of related machines. Furthermore, stronger hardness of the UMPS problem implies better (almost optimal) hardness results for the related machines scheduling problem. 

Finally, we note that Conjecture \ref{conj:dag-main1} implies that precedence constrained scheduling (even with communication delay = 0) is very hard to approximate when generalized to restricted assignment setting or unrelated machines.

Our confidence in the above conjectures stems from the fact that existing techniques, both the classical jobshops algorithms~\cite{LeightonMR94} and the recent LP hierarchies based algorithms  \cite{MRSSV,SchedulingWithCommunicationDelaysDKRTZ2020} fail to give non-trivial approximation guarantees for the UMPS problem. Furthermore, a candidate hard instances for the problem are \textit{layered} instances, where there are precedences between jobs $j_1 \prec j_2$ only if $j_1$ can be scheduled on the machine $i$ and $j_2$ can be scheduled on the machine $i+1$. 
These layered instances are closely related to the $k$-partite partitioning hypothesis of~\cite{BazziN15} and the integrality gap instances~\cite{MRSSV} for the scheduling with uniform communication delay problem. 

\subsection{A Brief History of Communication Delay Problem. }
\label{sec:relatedwork}

In this subsection, we give a brief overview of the literature on the scheduling with communication delay problem. 

\paragraph{Scheduling with precedences.} The scheduling with precedences to minimize makespan $\p|\Prec|\Cmax$ is a classical combinatorial optimization problem and is a special case of the communication delay problem with $c = 0$ for all pairs of jobs. In one of the earliest results in the scheduling theory, Graham's list scheduling algorithm~\cite{GrahamListScheduling1966} is a $2$-factor approximation algorithm for the problem. Recently, Svensson~\cite{Svensson10} gave a matching hardness of approximation result assuming (a variant of) the Unique Games Conjecture~\cite{BansalK09}. When the number of machines is a constant, a series of recent works have obtained $1+\epsilon$ approximation in nearly quasi-polynomial time \cite{LeveyR16,Garg18, KulkarniLTY20, li2021towards}.

\paragraph{Uniform communication delay setting.} The problem becomes much harder with communication delays, even when all the communication delays between the jobs are equal. This problem is denoted by $\p \mid \Prec, c \mid \Cmax$ is referred to as scheduling with uniform communication delays. 
In this setting, Graham's list scheduling algorithm obtains a $(c+1)$-factor approximation. This was improved to $2/3 \cdot (c+1)$ factor by Giroudeau et al.~\cite{GiroudeauKMP08} when the jobs have unit lengths ($\p \infty \mid \Prec, p_j=1, c \ge 2 \mid \Cmax$).
In recent concurrent and independent works, poly-logarithmic factor approximation algorithms have been obtained for the uniform communication delays problem $\p \mid \Prec, c \mid \Cmax$ by Maiti et al.~\cite{MRSSV} and Davies et al.~\cite{SchedulingWithCommunicationDelaysDKRTZ2020, DaviesKRTZ21}. 

On the hardness front, when $c=1$, Hoogeveen, Lenstra and Veltman~\cite{HoogeveenLV94} showed that the problem $\p \infty \mid \Prec, p_j=1, c=1 \mid \Cmax$ is NP-hard to approximate to a factor  better than $7/6$.
The result has been generalized for $c \ge 2$ to $(1+1/(c+4))$-hardness~\cite{GiroudeauKMP08}.\footnote{
	Papadimitriou and Yannakakis~\cite{PapadimitriouY90} claim a $2$-hardness for
	$\p \infty \mid \Prec, p_j=1, c \mid \Cmax$,
	but give no proof.
	Schuurman and Woeginger~\cite{SW99a} remark that ``it would be nice to have a proof for this claim''.
}

\paragraph{Scheduling with non-uniform communication delay.} We do not know of any algorithm for non-uniform communication delays $\p \infty \mid \Prec, c_{jk} \mid \Cmax$ problem.  
On the hardness side, the best hardness is the above small constant hardness of the uniform communication delay setting. 
While our main result shows that logarithmic hardness for this problem, it is conceivable that it admits a polylog approximation algorithm, although our conjectures suggest otherwise.

\paragraph{Duplication model.}
Scheduling with communication delays problem has also been studied in the duplication model where we allow jobs to be duplicated (replicated), i.e., executed on more than one machine to avoid communication delays. In this easier model, for the general $\p \infty \mid \Prec, p_j, c_{jk}, \dup \mid \Cmax$ problem, there is a simple  $2$-factor approximation algorithm by Papadimitriou and Yannakakis~\cite{PapadimitriouY90}. 
On the other hand, \cite{PapadimitriouY90} also show the NP-hardness of
$\p \infty \mid \Prec, p_j=1, c, \dup \mid \Cmax$. 
Note that the $O(1)$ approximation algorithm for the problem with duplication is in sharp contrast to our hardness result (\Cref{thm:main}) illustrating that the problem is significantly harder without duplication. 

\subsection{Discussion and Open Problems}

While we make progress on the hardness of approximation for the scheduling with non-uniform communication delay problem, our main conceptual contribution of this work is initiating the formal study of the UMPS problem. When the jobs have lengths, the problem does not admit a polylogarithmic approximation. However, much less is known for the unit lengths case. We now mention a few open problems in this direction.

\begin{enumerate}
    \item The key open problem is to prove (or disprove)~\Cref{conj:dag-main}. A positive resolution of the conjecture would prove the hardness of scheduling related machines with precedences, a long-standing open problem in scheduling theory. By the same reduction as in the proof of~\Cref{thm:main},~\Cref{conj:dag-main} also implies a logarithmic hardness of approximation of non-uniform communication delay problem even when the jobs have unit lengths ($\p \infty \mid \Prec, p_j=1, c_{jk}\mid \Cmax$). 
    \item On the other hand, obtaining good approximation algorithms for the UMPS problem would be even more exciting. Is~\Cref{conj:dag-main1} true, or is there a $(\log n)^{O(1)}$ factor approximation algorithm for the problem when the jobs have unit length?
    
\end{enumerate}

\subsection{Organization}
The rest of the paper is organized as follows. We first formally define the UMPS problem and relate it to the jobshops problem in~\Cref{sec:job-dag}. We then use the hardness of the UMPS problem to prove~\Cref{thm:main} in~\Cref{sec:comm-delay}. Finally, in~\Cref{sec:conditional}, we show that ~\Cref{conj:dag-main} implies improved hardness of the related machine scheduling with precedences and that the hypothesis of~\cite{BazziN15} implies a super constant hardness of the UMPS problem with unit lengths. 

\section{Unique Machine Precedence Constraints Scheduling problem}
\label{sec:job-dag}
We first formally define the Unique Machine Precedence constraint Scheduling (UMPS) problem. 
\begin{definition}(Unique Machine Precedence constraint Scheduling)
In the Unique Machine Precedence constraint Scheduling (UMPS) problem, the input is a set of $m$ machines and $n$ jobs $j_1, j_2, \ldots, j_n$ with precedence relations among them. 
Furthermore, each job $j_l$ can be scheduled only on a fixed machine $M(l)\in [m]$, and takes $p(l)$ time to complete. The jobs should be scheduled non-preemptively i.e., once a machine starts processing a job $j_l$, it has to finish the job $j_l$ before processing other jobs.
The objective is to schedule the jobs on the corresponding machines in this non-preemptive manner respecting the precedence relations, to minimize the makespan. 
\end{definition}

We note that the UMPS problem is a generalization of the classical jobshops problem that we formally define below. 

\begin{definition}(Job shops)
In the jobshops problem, the input is a set of $n$ jobs to be processed on a set $M$ of $m$ machines. Each job $j$ consists of $\mu_j$ operations $O_{1,j}, O_{2,j}, \ldots, O_{\mu_j,j}$. Operation $O_{i,j}$ must be processed for $p_{i,j}$ units of time without interruptions on the machine $m_{i,j}\in M$, and can only be scheduled if all the preceding operations $O_{i',j}, i'<i$ have finished processing. The objective is to schedule all the operations on the corresponding machines to minimize the makespan. 
\end{definition}
Note that jobshops problem is a special case of the UMPS problem, corresponding to the case when the precedence DAG is a disjoint union of chains. 
The jobshops problem has received a lot of attention and played an important role in the development of key algorithmic techniques~\cite{lawler1993sequencing, LeightonMR94}. On the hardness front, Mastrolilli and Svensson showed almost optimal hardness results for the problem in a breakthrough result~\cite{MastrolilliS11}. 

\begin{theorem}
\label{thm:jobshops-hardness}
For every constant $\epsilon >0$, assuming $\emph{\textsf{NP}}\nsubseteq \emph{\textsf{ZTIME}}\left(n^{(\log n)^{O(1)}}\right)$, there is no polynomial time $(\log n)^{1-\epsilon}$ factor approximation algorithm for the jobshops problem, where $n$ is the total number of operations in the given jobshops instance. 
\end{theorem}

As a corollary, we obtain the following hardness result. 
\begin{corollary}
\label{thm:job-shops-DAG-weights}
For every constant $\epsilon >0$, assuming $\emph{\textsf{NP}}\nsubseteq \emph{\textsf{ZTIME}}\left(n^{(\log n)^{O(1)}}\right)$, there is no polynomial time $(\log n)^{1-\epsilon}$ factor approximation algorithm for the UMPS problem. 
\end{corollary}

\section{Hardness of Scheduling With Non-Uniform Communication Delay}
\label{sec:comm-delay}

We now give a reduction from the UMPS problem to the non-uniform communication delay problem, thereby proving the hardness of the non-uniform communication delay problem. 

\main*

\medskip \noindent \textbf{Reduction. }
Let $I$ be the instance of the UMPS DAG problem with $n$ jobs $j_1, j_2, \ldots, j_n$, and $m$ machines. Furthermore, each job $j_l$ has a processing time $p(l)$ and can be scheduled only on the machine $M(l) \in [m]$. For an index $i \in [m]$, let $J(i) \subseteq \{j_1, j_2, \ldots, j_n\}$ denote the set of jobs that can be scheduled on the machine $i$. 

Roughly speaking, our idea in the reduction is to output a non-uniform communication delay instance where we force the jobs in $J(i)$ to be scheduled on the same machine, for every $i \in [m]$. We achieve this by adding a set of $m$ dummy jobs $j^*_1, j^*_2, \ldots, j^*_m$ and adding precedences with very large communication delay from all the jobs in $J(i)$ to $j^*_i$ for every $i \in [m]$.
More formally, we define an instance $I'$ of the non-uniform communication delay problem as follows. First, we choose a large integer $C_{\infty} = n \sum_{l=1}^n p(l)$. There are $n+m$ jobs in $I'$: a set of $n$ jobs $j'_1, j'_2, \ldots, j'_n$ such that for each $l \in [n]$, the processing time of $j'_l$ is equal to $p(l)$, and a set $\{ j^*_1, j^*_2, \ldots, j^*_m\}$ of $m$ dummy jobs, each with processing time $1$. For every precedence relation $j_u \prec j_v$ in the original instance $I$, there is a precedence relation $j'_u \prec j'_v$ in $I'$ with communication delay $0$. Finally, for every $i$, and every job $j_l \in J(i)$, there is a precedence relation $j'_l \prec j^*_i$ with communication delay $C_{\infty}$. 

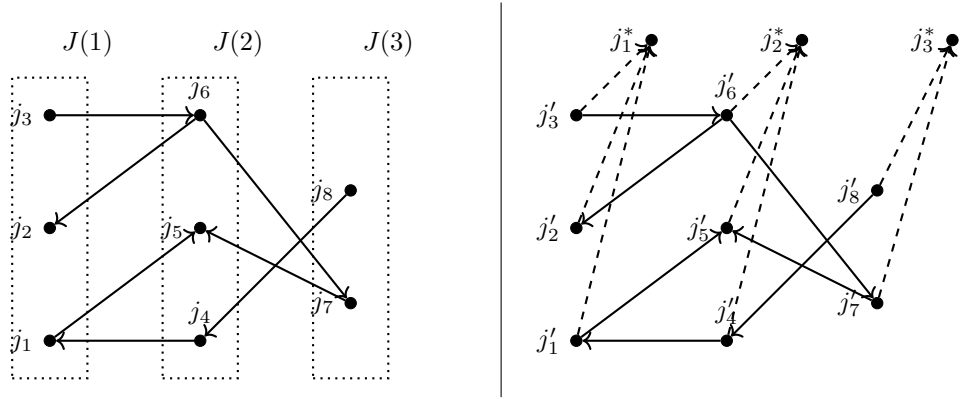
\begin{figure}
    \centering
    \begin{tikzpicture}
    
\draw (0,0.5)  node[fill,circle, draw, fill=black,inner sep=0,minimum size=0.15cm,label={180:\small $j_1$}](A) {};
\draw (0,2)  node[fill,circle, draw, fill=black,inner sep=0,minimum size=0.15cm,label={180:\small $j_2$}](B) {};
\draw (0,3.5)  node[fill,circle, draw, fill=black,inner sep=0,minimum size=0.15cm,label={180:\small $j_3$}](C) {};
\draw (2,0.5)  node[fill,circle, draw, fill=black,inner sep=0,minimum size=0.15cm,label={90:\small $j_4$}](D) {};
\draw (2,2)  node[fill,circle, draw, fill=black,inner sep=0,minimum size=0.15cm,label={180:\small $j_5$}](E) {};
\draw (2,3.5)  node[fill,circle, draw, fill=black,inner sep=0,minimum size=0.15cm,label={90:\small $j_6$}](F) {};
\draw (4,1)  node[fill,circle, draw, fill=black,inner sep=0,minimum size=0.15cm,label={180:\small $j_7$}](G) {};
\draw (4,2.5)  node[fill,circle, draw, fill=black,inner sep=0,minimum size=0.15cm,label={180:\small $j_8$}](H) {};
    
    \draw [->, thick] (D) -- (A);
    \draw [->, thick] (A) -- (E);
    \draw [->, thick] (F) -- (B);
    \draw [->, thick] (C) -- (F);
    \draw [->, thick] (F) -- (G);
    \draw [->, thick] (G) -- (E);
    \draw [->, thick] (H) -- (D);

\draw[dotted, thick] (-0.5,0) rectangle (0.5,4) node[label={\small $J(1)$}] {};
\draw[dotted, thick] (1.5,0) rectangle (2.5,4) node[label={\small $J(2)$}] {};
\draw[dotted, thick] (3.5,0) rectangle (4.5,4) node[label={\small $J(3)$}] {};

\draw (6,-0.25) -- (6,5);

\draw (7,0.5)  node[fill,circle, draw, fill=black,inner sep=0,minimum size=0.15cm,label={180:\small $j'_1$}](A1) {};
\draw (7,2)  node[fill,circle, draw, fill=black,inner sep=0,minimum size=0.15cm,label={180:\small $j'_2$}](B1) {};
\draw (7,3.5)  node[fill,circle, draw, fill=black,inner sep=0,minimum size=0.15cm,label={180:\small $j'_3$}](C1) {};
\draw (9,0.5)  node[fill,circle, draw, fill=black,inner sep=0,minimum size=0.15cm,label={90:\small $j'_4$}](D1) {};
\draw (9,2)  node[fill,circle, draw, fill=black,inner sep=0,minimum size=0.15cm,label={180:\small $j'_5$}](E1) {};
\draw (9,3.5)  node[fill,circle, draw, fill=black,inner sep=0,minimum size=0.15cm,label={90:\small $j'_6$}](F1) {};
\draw (11,1)  node[fill,circle, draw, fill=black,inner sep=0,minimum size=0.15cm,label={180:\small $j'_7$}](G1) {};
\draw (11,2.5)  node[fill,circle, draw, fill=black,inner sep=0,minimum size=0.15cm,label={180:\small $j'_8$}](H1) {};

\draw (8,4.5)  node[fill,circle, draw, fill=black,inner sep=0,minimum size=0.15cm,label={180:\small $j^*_1$}](L) {};
\draw (10,4.5)  node[fill,circle, draw, fill=black,inner sep=0,minimum size=0.15cm,label={180:\small $j^*_2$}](M) {};
\draw (12,4.5)  node[fill,circle, draw, fill=black,inner sep=0,minimum size=0.15cm,label={180:\small $j^*_3$}](N) {};

    \draw [->, thick] (D1) -- (A1);
    \draw [->, thick] (A1) -- (E1);
    \draw [->, thick] (F1) -- (B1);
    \draw [->, thick] (C1) -- (F1);
    \draw [->, thick] (F1) -- (G1);
    \draw [->, thick] (G1) -- (E1);
    \draw [->, thick] (H1) -- (D1);

    \draw [->,dashed, thick] (A1) -- (L);
    \draw [->,dashed, thick] (B1) -- (L);
    \draw [->,dashed, thick] (C1) -- (L);
    \draw [->,dashed, thick] (D1) -- (M);
    \draw [->,dashed, thick] (E1) -- (M);
    \draw [->,dashed, thick] (F1) -- (M);
    \draw [->,dashed, thick] (G1) -- (N);
    \draw [->,dashed, thick] (H1) -- (N);
    
    \end{tikzpicture}
    \caption{Illustration of the reduction from non-preemptive jobshops on DAG to non-uniform communication delays. In the communication delay instance on the right, the dashed arrow precedences have communication delay $C_{\infty}$ while the normal arrow precedences have communication delay $0$.}
    \label{fig:main-reduction}
\end{figure}

\medskip \noindent \textbf{Completeness.}
Suppose that there is a scheduling of $I$ with makespan at most $L$. Then, we claim that there is a scheduling of $I'$ with makespan at most $L+1$. 
We use $m$ machines and schedule the job $j'_l$ on the machine $M(l)$ in the same time slot used by the scheduling of $I$. 
As the communication delay of the precedences among the jobs $\{ j'_1, j'_2, \ldots, j'_n \}$ is zero, we can schedule these jobs using $m$ machines with makespan at most $L$. 
Now, after all the jobs $\{j'_1, j'_2, \ldots, j'_n\}$ have been scheduled, we schedule the job $j^*_i$ in the machine $i$, for every $i \in [m]$. As we are scheduling all the jobs in $J(i)$ and $j^*_i$ on the same machine for every $i \in [m]$, we incur no communication delay when we are scheduling the dummy jobs, and we can schedule all the dummy jobs $j^*_1, j^*_2, \ldots, j^*_m$ simultaneously in the time slot between $L$ and $L+1$. 

\medskip \noindent \textbf{Soundness.} Suppose that there is a scheduling of $I'$ with makespan at most $L$. Then, we claim that there is a scheduling of $I$ with makespan at most $L$ as well. 

Note that there is a trivial scheduling of $I$ where we schedule each job one by one after topologically sorting them, that has a makespan of $\sum_{j=1}^n p(j)$. Thus, henceforth, we assume that $L \leq \sum_{j=1}^n p(j)$. 
For an index $i \in [m], $ let $J'(i)$ be the subset of jobs in $I'$ whose corresponding jobs in $I$ are to be scheduled on the machine $i$. 
\[
J'(i) := \{ j'_l : M(l)=m \}
\]
We claim that in the scheduling of $I'$ with makespan at most $L$, for every $i \in [m]$, all the jobs in $J'(i)$ must be scheduled on the same machine. 
Suppose for the sake of contradiction that this is not the case. If there are jobs $j'_{l_1}$ and $j'_{l_2}$ such that $M(l_1)=M(l_2)=i$ are scheduled on different machines $i_1, i_2$ in the scheduling of $I'$, at least one of $j'_{l_1}$ and $j'_{l_2}$ is scheduled on a different machine from that of $j^*_i$. 
However, as there are precedence relations $j'_{l_1} \prec j^*_i$ and $j'_{l_2} \prec j^*_i$ with communication delay $C_{\infty}$, at least one of the precedence relations has to wait for the communication delay, and thus, the makespan is at least $C_{\infty}>L$, a contradiction.  

Thus, for every $i \in [m]$, all the jobs in $J'(i)$ are processed on the same machine in $I'$.
This implies that at any point of time, at most one job from $J'(i)$ is being processed, for every $i \in [m]$. Using this observation, we output a scheduling of $I$: for every job $j_l \in J(i)$, we schedule $j_l$ in the same time slot used by the job $j'_l$ in the scheduling of $I'$. By the above observation, every machine $i \in [m]$ is used at most once at any time point. Furthermore, as the scheduling of $I'$ respects the precedence conditions, the new scheduling of $I$ also respects the precedence conditions. 
Note that the makespan of this scheduling of $I$ is equal to $L$. 
This completes the proof that there exists a scheduling of $I$ with makespan at most $L$, if there exists a scheduling of $I'$ with makespan at most $L$. 

\medskip 
This completes the proof of ~\Cref{thm:main}. We remark that the same reduction also proves a $(\log n)^{1-\epsilon}$ factor inapproximability of the finite machine version $\p \mid \Prec, c_{jk} \mid \Cmax$ of the non-uniform communication delay problem. 
\section{Conditional Hardness of Scheduling With Precedence Constraints on Related Machines}
\label{sec:conditional}
In this section, we first prove that~\Cref{conj:dag-main} implies improved hardness of scheduling related machines with precedences. 

We begin by formally defining the scheduling related machines with precedences problem ($\q \mid \Prec \mid \Cmax$). 
\begin{definition}(Scheduling related machines with precedences)
In the scheduling related machines with precedences problem, the input is a set of $m$ machines $\mathcal{M}$ and a set of $n$ jobs $\mathcal{J}$ with precedences among them. Furthermore, each machine $i$ has speed $s_i \in \mathbb{Z}^+$, and each job $j$ has processing time $p_j \in \mathbb{Z}^+$, and scheduling the job $j$ on machine $i$ takes $\frac{p_j}{s_i}$ units of time. The objective is to schedule the jobs on the machines non-preemptively respecting the precedences constraints, to minimize the makespan.   
\end{definition}
An algorithm with $O(\log m)$ approximation guarantee for the problem was given independently by Chudak and Shmoys~\cite{UniformlyRelatedMachinesWithPrecedences-ChudakShmoys-JALG99}, and Chekuri and Bender~\cite{ChekuriB01}. On the hardness side, a hardness factor of $2$ follows from the identical machines setting~\cite{Svensson10}, assuming a variant of the Unique Games Conjecture. Furthermore, Bazzi and Norouzi-Fard~\cite{BazziN15} put forth a hypothesis on the hardness of a $k$-partite graph partitioning problem, which implies a super constant hardness of the scheduling related machines with precedences problem.

We now prove that ~\Cref{conj:dag-main} implies poly logarithmic hardness of scheduling related machines with precedences problem. 

\related*

\medskip \noindent \textbf{Reduction.} 
Our reduction is essentially the same reduction as in~\cite{BazziN15} where the authors obtained conditional hardness of the related machine scheduling with precedences problem assuming the hardness of a certain $k$-partite graph partitioning problem. However, our soundness analysis needs more technical work.

We start with an instance $I$ of the UMPS problem with $n$ unit sized jobs $j_1, j_2, \ldots,j_ n$ and $m$ machines, and every job $j_l$ can only be scheduled on the machine $M(l) \in [m]$. Furthermore, we let $J(i)\subseteq [n], i \in [m]$ denote the set of all the jobs that can be scheduled on the machine $i$. 

We now output an instance $I'$ of the related machine scheduling problem. 
We choose a parameter $\kappa=10n^3m$.
For every $l \in [n]$, we have a set $\mathcal{J}_{l}$ of $\kappa^{2(m-M(l))}$ jobs in $I'$. The processing time of each of these jobs is equal to $\kappa^{M(l)-1}$. For every $i \in [m]$, we have $\mathcal{M}_i$, a set of $\kappa^{2(m-i)}$ machines, each with speed $\kappa^{i-1}$. Furthermore, for every precedence constraint $j_u \prec j_v$ in $I$, we have $j'_{l_1} \prec j'_{l_2}$ for every $j'_{l_1} \in \mathcal{J}_u$ and $j'_{l_2} \in \mathcal{J}_v$. 

\medskip \noindent \textbf{Completeness.} Suppose that there is a scheduling of $I$ with makespan equal to $L$. Then, we claim that there is a scheduling of $I'$ with makespan at most $L$ as well. 
Note that all the jobs in $\mathcal{J}_l$ can be scheduled on the machines $\mathcal{M}_{M(l)}$ in unit time. 
We obtain a scheduling of $I'$ by assigning the jobs $\mathcal{J}_{l}$ to the machines $\mathcal{M}_{M(l)}$ in the time slot used in $I$ to schedule the job $j_l$. 
This scheduling of $I'$ is indeed a valid scheduling, and has a makespan of at most $L$.

\medskip \noindent \textbf{Soundness.} We prove the soundness part in the lemma below. 

\begin{lemma}
Suppose that there is a scheduling of $I'$ with makespan $L$. Then, we will show that there is a scheduling of $I$ with makespan at most $2L$. 
\end{lemma}

\begin{proof}
Note that there is a trivial scheduling of $I$ where we schedule jobs in a topological sort one by one, with makespan equal to $n$. Thus, henceforth, we assume that $L \leq n$. 

Let $\gamma = \frac{1}{10n^2}$. 
We claim that for every $l \in [n]$, at most $\gamma \kappa^{2(m-M(l))}$ jobs in $\mathcal{J}_l$ are processed by machines that do not belong to $\mathcal{M}_{M(l)}$ in the scheduling $I'$. The proof of this claim follows from Lemma 1 of~\cite{BazziN15}, and we present it here for the sake of completeness.
Fix an index $l \in [n]$, and for ease of notation, let $i =M(l)$.
First, as each job in $\mathcal{J}_l$ has length $\kappa^{i-1}$, and the processing speed of each machine in $\mathcal{M}_j, j <i$ is at most $\kappa^{j-1} \leq \kappa^{i-2}$, no job in $\mathcal{J}_l$ is scheduled on machines in $\mathcal{M}_j$, $j <i$, as the makespan of $I'$ is at most $n < \kappa$. 
Now, consider an integer $j \in [m]$, $j >i$. There are $\kappa^{2(m-j)}$ machines in $\mathcal{M}_j$, and they have a processing speed of $\kappa^{j-1}$. Thus, in time $L \leq n$, they can process at most 
\[
\frac{n \cdot \kappa^{2(m-j)} \cdot \kappa^{j-1}}{\kappa^{i-1}} \leq \frac{n}{\kappa} \cdot \kappa^{2(m-i)}
\]
jobs of $\mathcal{J}_l$. Taking union over all $j >i$, we get that at most \[
\frac{nm}{\kappa} \cdot \kappa^{2(m-i)} \leq \frac{1}{10n^2} \kappa^{2(m-i)}
\]
jobs in $\mathcal{J}_l$ are processed by machines outside $\mathcal{M}_i$. 
In other words, for every job $j_l$ of $I$, at most $\gamma$ fraction of the jobs in $\mathcal{J}_l$ are processed by machines outside $\mathcal{M}_{M(l)}$. 

Now, consider a scheduling of the jobs in $I'$ where for every $l\in [n]$, we get rid of the jobs in $\mathcal{J}_l$ that are processed by machines outside $\mathcal{M}_{M(l)}$.
After removing the jobs processed by other machines, we still have that for every $l \in [n]$, at least $1-\gamma$ fraction of the jobs in $\mathcal{J}_{l}$ are processed. Also observe that since we are only deleting some jobs, the makespan of the new scheduling is at most $L$ as well.
Recall that processing each job in $\mathcal{J}_{l}$ takes unit time on the machines in $\mathcal{M}_{M(l)}$.

Using this observation, we obtain a fractional scheduling of $I$ in time $L$ as follows.
For every $l \in [n]$ and $t \in [L]$, define the variable $x_{l,t}$ to be the fraction of the jobs of $\mathcal{J}_l$ that are scheduled by the machines $\mathcal{M}_{M(l)}$ in the time slot $t$. 
By the above discussion, we get the following properties of this fractional scheduling. 
\begin{enumerate}
    \item Every job $l\in [n]$ is almost fully processed. For every $l \in [n]$, we have 
    \[
    \sum_{t=1}^L x_{l,t} \geq 1-\gamma
    \]
    \item Every machine is used only for processing a single unit of job in a time slot. 
    \[
    \sum_{l \in J(i)} x_{l,t} \leq 1\,\,\forall i \in [m],t\in[L].
    \]
    \item If there is a precedence constraint $j_{l_1} \prec j_{l_2}$ in $I$, $l_2$'s processing is done only in the time slots after $l_1$ is fully processed. More formally, 
    \[
    x_{l_1,t} >0 \Rightarrow x_{l_2, t'} = 0 \, \forall t' \leq t
    \]
\end{enumerate}

We will now show that the fractional scheduling implies that the instance $I$ has an integral scheduling with makespan at most $O(L)$, thereby proving the Lemma. 
We will prove this in two steps: first, we modify the fractional scheduling to obtain another fractional scheduling with better structure, and then next, we use this to obtain the integral scheduling. 

For a job $l \in [n]$, define the starting time $t^s_l$ and the end time $t^e_l$ as the minimum and the maximum times at which $l$ is being processed. 
\[
t^s_l = \min\{t : x_{l,t}>0\} \,,\, t^e_l = \max\{t:x_{l,t}>0\}
\]
Note that if we have $j_{l_1}\prec j_{l_2}$, $t^s_{l_2}>t^e_{l_1}$. 
We now modify the fractional scheduling to ensure that each machine processes the job with the lowest ending time first, from the available set of the jobs. 
More formally, for a machine $i \in [m]$, consider the pair of jobs $l_1, l_2 \in J(i)$ and time slot $t \in [L]$ satisfying the following conditions.
\begin{enumerate}
    \item[(C1)] The job $l_1$ has lower ending time: $t^e_{l_1}<t^e_{l_2}$, or $t^e_{l_1}=t^e_{l_2}$ and $l_1 < l_2$. 
    \item[(C2)] The job $l_1$ can be processed on the time slot $t$, but the job $l_2$ is processed instead of finishing the job $l_1$: 
    \[
    t^s_{l_1} \leq t < t^e_{l_1}\,,\, x_{l_2,t}>0
    \]
\end{enumerate}
If there are jobs $l_1, l_2$ and time slot $t$ satifying these conditions, we swap the processing times, and process the job $l_1$ in the time slot $t$ instead of $l_2$. More formally, let $t'>t$ be such that $x_{l_1,t'}>0$. Let $y=\min(x_{l_1,t'}, x_{l_2,t})$. We obtain a new fractional scheduling by setting 
\begin{align*}
x_{l_1, t'} = x_{l_1, t'} - y \, &, \, x_{l_1, t} = x_{l_1,t} + y \\ 
x_{l_2, t'} = x_{l_2, t'} + y \, &, \, x_{l_2, t} = x_{l_2,t} - y 
\end{align*}
Note that the operation does not increase the ending time of either job and does not decrease the starting time of either job and thus, results in a valid fractional scheduling respecting the precedence conditions. 
We repeat the swapping operations until there is no triple $i,j,t$ left where both (C1) and (C2) are true.
We also update the starting and ending times of the jobs $t^s_l$ and $t^e_l$ appropriately when we apply the swapping operations.  

Next, we apply another transformation to the fractional scheduling by filling the empty slots in the machines, if there are any. 
More formally, consider a time slot $t \in [L]$ and job $l \in [n]$ such that the following hold.
\begin{enumerate}
    \item[(D1)] The time slot $t$ is not fully utilized: 
    \[
    \sum_{l' \in J(M(l))} x_{l',t} <1
    \]
    \item[(D2)] The job $l$ can be scheduled on the time slot $t$ instead of leaving the machine idle: $t^s_l \leq t < t^e_l$. 
\end{enumerate}
If there is a time slot $t$ and job $l$ such that the above two conditions hold, we fill the empty slot in the time slot $t$ by processing the job $l$. Let $t'>t$ be such that $x_{l,t'}>0$. 
Let $y = \min ( x_{l,t'}, 1-\sum_{l' \in J(M(l))} x_{l',t})$. 
We set 
\[
x_{l,t} = x_{l,t} + y \,, \, x_{l,t'} = x_{l,t'}-y
\]
We repeat these operations iteratively until no empty slots can be filled. Similar to the previous case, we update the starting and ending times of the jobs appropriately. 

After the two types of operations, we obtain a fractional scheduling with the following property: at every time slot $t$, for a machine $i \in [m]$, let $S_{i,t}$ be the set of jobs that can be scheduled on $i$ in the time slot $t$:
\[
S_{i,t} := \{ l \in J(i) : t^s_l \leq t \leq t^e_l \}
\]
We sort the jobs in $S_{i,t}$ as $\{ l_1, l_2, \ldots, l_k\}$ by increasing order of ending times, and breaking ties based on the index. 
The fractional scheduling greedily schedules the jobs $l_1, l_2, \ldots, $ in that order. More formally, we have 
\[
x_{l_1, t} = \sum_{t'=1}^L x_{l_1,t'} - \sum_{t'=1}^{t-1}x_{l_1,t'}
\]
and 
\[
x_{l_2,t} = \min \Big(1-x_{l_1,t}, \sum_{t'=1}^L x_{l_2,t'} - \sum_{t'=1}^{t-1}x_{l_2,t'}\Big)
\]
and so on. 

Our goal is to show that in this final fractional scheduling that we obtained, each machine schedules at most two jobs in any time slot. In order to prove this, we first define the following parameter, $P_{i,t}$, the amount of jobs \emph{partially} completed in the machine $i$ by the time $t$. 
\[
P_{i,t} = \sum_{l \in J(i): t^e_l>t} \,\, \sum_{ t'=1}^t  x_{l,t'}
\]
We claim that for every $i \in [m], t \in [L]$, we have $P_{i,t} \leq \gamma t$. Fix a machine $i \in [m]$. 
We will prove the claim by induction on $t$. 
\begin{enumerate}
    \item Base case when $t=1$. If no job is processed by the machine $i$ in the time slot $t=1$, the claim is trivially satisfied. 
    Else, let $l_1$ be the job in $J(i)$ with the lowest ending time, breaking ties by the lowest index. Note that the fractional scheduling fully schedules the job $l_1$ in the time slot $t=1$. As each job is processed for at least $1-\gamma$ duration, we get that $P_{i,1}$ is at most $\gamma$. 
    \item Inductive proof. Suppose that the claim holds for all $t' \leq t$ and consider the time slot $t+1$. For ease of notation, let $S = S_{i,t+1}$ be the set of jobs that can be processed on the machine $i$ in the time slot $t+1$.
    If $S$ is empty, the inductive claim trivially holds. Else, let $l \in S$ be the job with the lowest ending time (breaking ties by the least index). Note that the modified fractional scheduling finishes the job $l$ in the time slot $t+1$. Let $x'_{l,t}$ denote the amount of the job $l$ that is processed by time $t$ i.e., 
    $x'_{l,t}= \sum_{t'=1}^t x_{l,t}$. 
    The amount of jobs that are partially finished at the end of time slot $t+1$ is at most 
    \begin{align*}
    P_{i,t+1} &\leq P_{i,t} - x'_{l,t} + (1-x_{l,t+1}) \\ 
    &\leq P_{i,t} + \gamma \leq (t+1)\gamma
    \end{align*}
\end{enumerate}

We will now show that every machine processes at most $2$ jobs in a time slot. Consider a machine $i \in [m]$ and time slot $t\in [L]$. Let $S_{i,t} := \{l_1, l_2, \ldots, l_k\}$.
By the previous claim, we know that at most $(t-1)\gamma$ portion of the job $l_u$ is finished before time $t$, for every $u \in [k]$. Note that $(t-1)\gamma \leq L \gamma \leq \frac{1}{10 n}$. Thus, the greedy fractional scheduling can only schedule at most two jobs, as each of them takes at least $1-\frac{1}{10n}$ time. 
Finally, using this observation, we can duplicate every time slot to obtain an integral scheduling of $I$ with makespan at most $2L$. 
\end{proof}

\smallskip \noindent \textbf{Parameter analysis.} The number of machines in the related machines scheduling instance is $M = \kappa^{O(m)} = n^{O(m)}$, while the hardness gap is $(\log n)^{1-\epsilon'}$ for every $\epsilon'>0$. By setting $\epsilon'$ appropriately, we get a hardness of $(\log M)^{\Omega(1)}$ for the scheduling related machines with precedences problem.

\subsection{Hypothesis of~\cite{BazziN15} implies super-constant hardness of UMPS problem with unit lengths}

Bazzi and Norouzi-Fard~\cite{BazziN15} introduced the following hypothesis and proved that it implies a super-constant hardness of scheduling related machines with precedences. 

\begin{hypothesis}(\cite{BazziN15})
\label{hyp:bazzi-fard}
For every $\epsilon, \delta >0$ and constant integers $k, Q >0$, the following problem is NP-hard. Given a $k$-partite graph $G=(V_1, V_2, \ldots, V_k, E_1, E_2, \ldots, E_{k-1})$ with $|V_i|=n$ for all $1 \leq i \leq k$, and $E_i$ being the set of edges between $V_i$ and $V_{i+1}$ for every $1 \leq i < k$, distinguish between the two cases: 
\begin{enumerate}
    \item (YES case) Every $V_i$ can be partitioned into $V_{i,0}, V_{i,1}, \ldots, V_{i,Q-1}$ such that 
    \begin{itemize}
        \item There is no edge between $V_{i,j_1}$ and $V_{i+1,j_2}$ for all $1 \leq  i < k, j_1 > j_2 \in [Q]$. 
        \item $|V_{i,j}| \geq \frac{(1-\epsilon)}{Q}n$ for all $i \in [k], j \in [Q]$. 
    \end{itemize}
    \item (NO case) For every $1 <i \leq k$, and any two sets $S, T$ with $S \subseteq V_i$, $T \subseteq V_{i-1}$, $|S|=|T|=\delta n$, there is an edge between $S$ and $T$.  
\end{enumerate}

\end{hypothesis}

We now prove that the above hypothesis implies that it is NP-hard to obtain a constant factor approximation algorithm for the UMPS problem, even when all the jobs have unit length. 

\smallskip \noindent \textbf{Reduction.} Given an instance of $k$-partite problem $I$, we output an instance $I'$ of the UMPS problem as follows: there are $n'=nk$ unit sized jobs in $I'$, one job corresponding to each vertex of $G$. There are $k$ machines, and all the jobs in $V_i, i \in [k]$ can only be scheduled on the machine $i$. For every edge $e =(u,v)$ in the graph such that $u \in V_{i}, v \in V_{i+1}$, we have a precedence condition $u \prec v$ in $I'$. We choose the parameter $Q=k$, and $\delta = \epsilon=\frac{1}{k}$.

\smallskip \noindent \textbf{Completeness.} Suppose that the YES case of~\Cref{hyp:bazzi-fard} holds i.e., there is a partition of $V_i$ into $V_{i,0}, V_{i,1}, \ldots, V_{i,Q-1}$ respecting the two conditions above. Then, we claim that there is a scheduling of $I'$ with makespan at most $3n$. 
For every machine $i \in [k]$, we schedule the jobs in $V_{i,0}$ (in arbitrary order), and then the jobs in $V_{i,1}$ (in arbitrary order) and so on. However, we start the execution of the jobs in $V_{i,0}$ at time $t_i$, and then, execute the jobs in $V_{i,l}$ immediately after the execution of the jobs in $V_{i,l-1}$ for all $l\geq 1$. 
The parameters $t_i$, $i \in [k]$ are chosen such that for every pair of jobs $u,v$ with $u \prec v$, $u$ is guaranteed to have scheduled before $v$. In particular, we choose $t_i = (i-1)n\left(\epsilon +\frac{1}{Q}\right)$.

We now prove that this results in a valid scheduling that respects the precedence conditions. 
Consider a pair of jobs $u, v$ such that $u \in V_{i}, v \in V_{i+1}$ such that $u \prec v$. As the $k$-partite graph satisfies the YES condition, we have integers $j_1, j_2$ such that $u \in V_{i,j_1}$ and $v \in V_{i+1,j_2}$, and $j_1 \leq j_2$. 
Note that $u$ is processed by time at most 
\[
t_u = t_i + |V_{i,0}|+|V_{i,1}|+\ldots +|V_{i,j_1}|
\]
Furthermore, $v$ is processed only after time 
\[
t_v = t_{i+1}+|V_{i+1,0}|+|V_{i+1,1}|+\ldots+|V_{i+1,j_2-1}|
\]
We have 
\begin{align*}
    t_v - t_u &= t_{i+1}+|V_{i+1,0}|+|V_{i+1,1}|+\ldots+|V_{i+1,j_2-1}| - (t_i + |V_{i,0}|+|V_{i,1}|+\ldots +|V_{i,j_1}|) \\ 
    &= n\left(\epsilon + \frac{1}{Q}\right) + |V_{i+1,0}|+|V_{i+1,1}|+\ldots+|V_{i+1,j_2-1}| - ( n-|V_{i,j_1+1}|+|V_{i,j_1+2}|+\ldots +|V_{i,Q-1}|) \\ 
    &\geq n\left(\epsilon + \frac{1}{Q}\right) + j_2\frac{(1-\epsilon)n}{Q} - \left( n - \frac{(Q-j_1-1)(1-\epsilon)n}{Q}\right) \\ 
    &\geq n \left( \epsilon + \frac{1}{Q}\right) + j_1\frac{(1-\epsilon)n}{Q} - \left( j_1\frac{(1-\epsilon)n}{Q} + \epsilon n + \frac{(1-\epsilon)n}{Q} \right) \geq 0
\end{align*}
Thus, the schedule is a valid scheduling of $I'$. 
The makespan of this scheduling is at most $t_k+n = (k-1)n\left( \epsilon + \frac{1}{Q} \right) + n \leq 3n$. 

\smallskip \noindent \textbf{Soundness.} Suppose that the NO case of~\Cref{hyp:bazzi-fard} holds. We claim that in this case, the makespan of $I'$ is at least $(1-2\delta)kn$. For every $i \in [k]$, let $s_i$ denote the time at which the machine $i$ has finished $(1-\delta)n$ jobs of $V_i$. For an index $i \in [k]$, let $S(i) \subseteq V_i$ denote the set of jobs that are not processed by the time $s_i$. By the definition of $s_i$, we have $|S(i)|\geq \delta n$. 
By the NO case of~\Cref{hyp:bazzi-fard}, we get that there are at least $(1-\delta)n$ jobs in $V_{i+1}$ that have dependencies in $S(i)$. Note that all these jobs can be scheduled only after $s_i$. 
Thus, we get 
\[
s_{i+1} \geq s_i + (1-2\delta)n  \,\,\forall i \in [k-1]
\]
Summing over all $i$, we get that the makespan of the scheduling is at least $(1-2\delta)kn$, which is at least $\frac{kn}{2}$ when $k \geq 4$. By choosing $k$ large enough, this completes the proof that assuming~\Cref{hyp:bazzi-fard}, it is NP-hard to obtain a $O(1)$ factor approximation algorithm for the UMPS problem when the jobs have unit lengths.

\bibliography{ref}
\bibliographystyle{alpha}
\end{document}